\documentclass[11pt]{article}
\usepackage[a4paper]{geometry}
\usepackage{amsfonts, amsmath, amssymb, amsthm, graphicx, caption, authblk, multirow, makecell, framed, float, xcolor, enumitem, tikz, hyperref}
\theoremstyle{definition}

\newtheorem{lemma}{Lemma}

\theoremstyle{remark}

\newcommand*{\mybox}[1]{%
  \framebox{\raisebox{0cm}[0.5\baselineskip][0.05\baselineskip]{%
    \hbox to 0.1cm {\hss#1\hss}}}\hspace{0.05cm}}

\begin{document}
\title{Physical ZKP for Connected Spanning Subgraph: Applications to Bridges Puzzle and Other Problems}
\author[1]{Suthee Ruangwises\thanks{\texttt{ruangwises@gmail.com}}}
\author[1]{Toshiya Itoh\thanks{\texttt{titoh@c.titech.ac.jp}}}
\affil[1]{Department of Mathematical and Computing Science, Tokyo Institute of Technology, Tokyo, Japan}
\date{}
\maketitle

\begin{abstract}
An undirected graph $G$ is known to both the prover $P$ and the verifier $V$, but only $P$ knows a subgraph $H$ of $G$. Without revealing any information about $H$, $P$ wants to convince $V$ that $H$ is a connected spanning subgraph of $G$, i.e. $H$ is connected and contains all vertices of $G$. In this paper, we propose an unconventional zero-knowledge proof protocol using a physical deck of cards, which enables $P$ to physically show that $H$ satisfies the condition without revealing it. We also show applications of this protocol to verify solutions of three well-known NP-complete problems: the Hamiltonian cycle problem, the maximum leaf spanning tree problem, and a popular logic puzzle called Bridges.

\end{abstract}

\section{Introduction}
A \textit{zero-knowledge proof (ZKP)} is an interactive protocol introduced by Goldwasser et al. \cite{zkp0}, which enables a prover $P$ to convince a verifier $V$ that a statement is correct without revealing any other information. A ZKP with perfect completeness and soundness must satisfy the following three properties.
\begin{enumerate}
	\item \textbf{Perfect Completeness:} If the statement is correct, then $V$ always accepts.
	\item \textbf{Perfect Soundness:} If the statement is incorrect, then $V$ always rejects.
	\item \textbf{Zero-knowledge:} During the verification, $V$ gets no extra information other than the correctness of the statement. Formally, there exists a probabilistic polynomial time algorithm $S$ (called a \textit{simulator}), without an access to $P$ but with a black-box access to $V$, such that the outputs of $S$ follow the same probability distribution as the outputs of the actual protocol.
\end{enumerate}

Goldreich et al. \cite{zkp} proved that a computational ZKP exists for every NP problem. Several recent results, however, instead considered an unconventional way of constructing ZKPs by using physical objects such as a deck of cards and envelopes. The benefit of these physical protocols is that they allow external observers to check that the prover truthfully executes the protocol (which is often a challenging task for digital protocols). They also have didactic values and can be used to teach the concept of ZKP to non-experts.

Consider a verification of the following condition. An undirected graph $G$ is known to both $P$ and $V$, but only $P$ knows a subgraph $H$ of $G$. Without revealing any information about $H$, $P$ wants to convince $V$ that $H$ is a connected spanning subgraph of $G$, i.e. $H$ is connected and contains all vertices of $G$.

A ZKP to verify the connected spanning subgraph condition is important because this condition is a part of many well-known NP-complete problems, such as the Hamiltonian cycle problem, the maximum leaf spanning tree problem, and a famous logic puzzle called \textit{Bridges}. To verify solutions of these problems, $P$ needs to show that his/her solution satisfies the connected spanning subgraph condition as well as some other conditions (which are relatively easier to show).

\subsection{Related Work}
Most of previous work in physical ZKPs aimed to verify a solution of popular logic puzzles: Sudoku \cite{sudoku0,sudoku}, Nonogram \cite{nonogram}, Akari \cite{akari}, Kakuro \cite{akari,kakuro}, KenKen \cite{akari}, Takuzu \cite{akari,takuzu}, Makaro \cite{makaro}, Norinori \cite{norinori}, Slitherlink \cite{slitherlink}, Juosan \cite{takuzu}, Numberlink \cite{numberlink}, Suguru \cite{suguru}, Ripple Effect \cite{ripple}, Nurikabe \cite{nurikabe}, and Hitori \cite{nurikabe}.

The theoretical contribution of these protocols is that they employ novel methods to physically verify specific functions. For example, a subprotocol in \cite{makaro} verifies that a number in a list is the largest one in that list without revealing any value in the list, and a subprotocol in \cite{sudoku0} verifies that a list is a permutation of all given numbers without revealing their order.

Some of these protocols can verify graph theoretic problems. For example, a protocol in \cite{numberlink} verifies a solution of the $k$ vertex-disjoint paths problem, i.e. a set of $k$ vertex-disjoints paths joining each of the $k$ given pairs of endpoints in a graph. In a recent work, a subprotocol in \cite{nurikabe} also verifies a condition related to connectivity. However, their protocol only works in a grid graph and also deals with a different condition from the one considered in this paper. (Their protocol only verifies that the selected cells on a board are connected together, not as a spanning subgraph of the whole board.)

\subsection{Our Contribution}
In this paper, we propose a physical card-based ZKP with perfect completeness and soundness to verify that a subgraph $H$ is a connected spanning subgraph of an undirected graph $G$ without revealing $H$.

We also show three possible applications of this protocol: verifying a Hamiltonian cycle in an undirected graph, verifying the existence of a spanning tree with at least $k$ leaves in an undirected graph, and verifying a solution of the Bridges puzzle.

\section{Preliminaries}
Each \textit{encoding card} used in our protocol has either $\clubsuit$ or $\heartsuit$ on the front side. All cards have indistinguishable back sides.

For $0 \leq x < k$, define $E_k(x)$ to be a sequence of consecutive $k$ cards, with all of them being \mybox{$\clubsuit$} except the $(x+1)$-th card from the left being \mybox{$\heartsuit$}, e.g. $E_3(0)$ is \mbox{\mybox{$\heartsuit$}\mybox{$\clubsuit$}\mybox{$\clubsuit$}} and $E_4(2)$ is \mbox{\mybox{$\clubsuit$}\mybox{$\clubsuit$}\mybox{$\heartsuit$}\mybox{$\clubsuit$}}. We use $E_k(x)$ to encode an integer $x$ in $\mathbb{Z}/k\mathbb{Z}$. This encoding rule was introduced by Shinagawa et al. \cite{polygon}.

The cards in $E_k(x)$ are arranged horizontally as defined above unless stated otherwise. In some situations, however, we may arrange the cards vertically, where the leftmost card becomes the topmost card and the rightmost card becomes the bottommost card.

In an $m \times k$ \textit{matrix} of cards, let Row $i$ denote the $i$-th topmost row and Column $j$ denote the $j$-th leftmost column.

\subsection{Pile-Shifting Shuffle} \label{shift}
A \textit{pile-shifting shuffle} on an $m \times k$ matrix shifts the columns of the matrix by a random cyclic shift, i.e. shifts the columns cyclically to the right by $r$ columns for a uniformly random $r \in \mathbb{Z}/k\mathbb{Z}$ unknown to all parties.

This protocol was developed by Shinagawa et al. \cite{polygon}. It can be implemented in real world by putting the cards in each column into an envelope and applying several \textit{Hindu cuts} to the sequence of envelopes \cite{hindu}.

\subsection{Sequence Selection Protocol} \label{select}
Suppose we have $k$ sequences $A_0$, $A_1$, ..., $A_{k-1}$, each encoding an integer in $\mathbb{Z}/m\mathbb{Z}$, and a sequence $B$ encoding an integer $b$ in $\mathbb{Z}/k\mathbb{Z}$. We propose the following \textit{sequence selection protocol}, which allows us to select a sequence $A_b$ (to be used as an input in other protocols) without revealing $b$.

\begin{figure}
\centering
\begin{tikzpicture}
\node at (0.0,-1.1) {$A_0$};
\node at (0.5,-1.1) {$A_1$};
\node at (1.0,-1.1) {...};
\node at (1.7,-1.1) {$A_{k-1}$};

\draw[->] (0.0,-0.4) -- (0.0,-0.8);
\draw[->] (0.5,-0.4) -- (0.5,-0.8);
\draw[->] (1.5,-0.4) -- (1.5,-0.8);

\node at (0.0,0.0) {\mybox{?}};
\node at (0.5,0.0) {\mybox{?}};
\node at (1.0,0.0) {...};
\node at (1.5,0.0) {\mybox{?}};

\node at (0.0,0.7) {\vdots};
\node at (0.5,0.7) {\vdots};
\node at (1.0,0.7) {\vdots};
\node at (1.5,0.7) {\vdots};

\node at (0.0,1.2) {\mybox{?}};
\node at (0.5,1.2) {\mybox{?}};
\node at (1.0,1.2) {...};
\node at (1.5,1.2) {\mybox{?}};

\node at (0.0,1.8) {\mybox{?}};
\node at (0.5,1.8) {\mybox{?}};
\node at (1.0,1.8) {...};
\node at (1.5,1.8) {\mybox{?}};

\node at (0,2.7) {\mybox{?}};
\node at (0.5,2.7) {\mybox{?}};
\node at (1,2.7) {...};
\node at (1.5,2.7) {\mybox{?}};
\draw[->] (1.8,2.7) -- (2.2,2.7);
\node at (2.4,2.7) {$B$};

\node at (0.0,3.3) {\mybox{?}};
\node at (0.5,3.3) {\mybox{?}};
\node at (1.0,3.3) {...};
\node at (1.5,3.3) {\mybox{?}};
\draw[->] (1.8,3.3) -- (2.2,3.3);
\node at (2.7,3.3) {$E_k(0)$};

\draw[] (-0.3,-0.3) -- (-0.3,4.7);
\draw[] (-2.1,3.7) -- (1.8,3.7);

\node at (-0.9,0.0) {$m+2$};
\node at (-0.6,0.7) {\vdots};
\node at (-0.6,1.2) {4};
\node at (-0.6,1.8) {3};
\node at (-0.6,2.7) {2};
\node at (-0.6,3.3) {1};
\node at (-1.6,1.65) {Row};

\node at (0.0,4.0) {1};
\node at (0.5,4.0) {2};
\node at (1.0,4.0) {...};
\node at (1.5,4.0) {$k$};
\node at (0.75,4.5) {Column};
\end{tikzpicture}
\caption{An $(m+2) \times k$ matrix $M$ constructed in Step 1}
\label{fig2}
\end{figure}
\begin{enumerate}
	\item Construct the following $(m+2) \times k$ matrix $M$ (see Fig. \ref{fig2}).
	\begin{enumerate}
		\item In Row 1, place a sequence $E_k(0)$. In Row 2, place the sequence $B$.
		\item In each Column $j=1,2,...,k$, place the sequence $A_{j-1}$ arranged vertically from Row 3 to Row $m+2$.
	\end{enumerate}
	\item Apply the pile-shifting shuffle to $M$.
	\item Turn over all cards in Row 2. Locate the position of a \mybox{$\heartsuit$}. Suppose it is at Column $j$.
	\item Select the sequence in Column $j$ arranged vertically from Row 3 to Row $m+2$. This is the sequence $A_b$ as desired. Turn over all face-up cards.
\end{enumerate}

After we are done using $A_b$ in other protocols, we can put $A_b$ back into $M$, apply the pile-shifting shuffle to $M$, then turn over all cards in Row 1 and shift the columns of $M$ cyclically such that the \mybox{$\heartsuit$} in Row 1 moves to Column 1. This reverts the matrix back to its original position, so we can reuse the sequences $A_0$, $A_1$, ..., $A_{k-1}$, and $B$.

\subsection{Enhanced Matrix}
In addition to the encoding cards, we also use \textit{marking cards}, each having a positive integer on the front side. All cards have indistinguishable back sides.

Starting from an $m \times k$ matrix of face-down encoding cards, place face-down marking cards \mybox{1}, \mybox{2}, ..., \mybox{$k$} from left to right on top of Row 1; this new row is called Row 0. Then, place face-down marking cards \mybox{2}, \mybox{3}, ..., \mybox{$m$} from top to bottom (starting at Row 2) to the left of Column 1; this new column is called Column 0. We call this structure an $m \times k$ enhanced matrix (see Fig. \ref{fig3}).

\begin{figure}
\centering
\begin{tikzpicture}
\node at (0.0,0.6) {\mybox{?}};
\node at (0.5,0.6) {\mybox{?}};
\node at (1.0,0.6) {\mybox{?}};
\node at (1.5,0.6) {\mybox{?}};
\node at (2.0,0.6) {\mybox{?}};

\node at (0.0,1.2) {\mybox{?}};
\node at (0.5,1.2) {\mybox{?}};
\node at (1.0,1.2) {\mybox{?}};
\node at (1.5,1.2) {\mybox{?}};
\node at (2.0,1.2) {\mybox{?}};

\node at (0.0,1.8) {\mybox{?}};
\node at (0.5,1.8) {\mybox{?}};
\node at (1.0,1.8) {\mybox{?}};
\node at (1.5,1.8) {\mybox{?}};
\node at (2.0,1.8) {\mybox{?}};

\node at (0.0,2.4) {\mybox{?}};
\node at (0.5,2.4) {\mybox{?}};
\node at (1.0,2.4) {\mybox{?}};
\node at (1.5,2.4) {\mybox{?}};
\node at (2.0,2.4) {\mybox{?}};

\node at (0.0,3.2) {\mybox{1}};
\node at (0.5,3.2) {\mybox{2}};
\node at (1.0,3.2) {\mybox{3}};
\node at (1.5,3.2) {\mybox{4}};
\node at (2.0,3.2) {\mybox{5}};
\node at (4.1,3.2) {(actually face-down)};

\node at (-0.7,0.6) {\mybox{4}};
\node at (-0.7,1.2) {\mybox{3}};
\node at (-0.7,1.8) {\mybox{2}};
\draw[->] (-0.7,0.0) -- (-0.7,0.3);
\draw[] (-0.7,0.0) -- (-0.2,0.0);
\node at (1.6,0.0) {(actually face-down)};

\draw[] (-1.0,0.3) -- (-1.0,4.7);
\draw[] (-2.5,3.6) -- (2.7,3.6);

\node at (-1.3,0.6) {4};
\node at (-1.3,1.2) {3};
\node at (-1.3,1.8) {2};
\node at (-1.3,2.4) {1};
\node at (-1.3,3.2) {0};
\node at (-2.0,1.5) {Row};

\node at (-0.7,3.9) {0};
\node at (0.0,3.9) {1};
\node at (0.5,3.9) {2};
\node at (1.0,3.9) {3};
\node at (1.5,3.9) {4};
\node at (2.0,3.9) {5};
\node at (1.0,4.4) {Column};
\end{tikzpicture}
\caption{An example of a $4 \times 5$ enhanced matrix}
\label{fig3}
\end{figure}

\subsection{Double-Scramble Shuffle} \label{shuffle}
In a \textit{double-scramble shuffle} on an $m \times k$ enhanced matrix, first rearrange Columns $1,2,...,k$ (including the marking cards in Row 0) by a uniformly random permutation unknown to all parties (which can be implemented by putting the cards in each column into an envelope and scrambling all envelopes together). Then, leave Row 1 as it is and rearrange Rows $2,3,...,m$ (including the marking cards in Column 0) by a uniformly random permutation unknown to all parties. This protocol was developed by Ruangwises and Itoh \cite{numberlink}.

\subsection{Rearrangement Protocol} \label{rearrange2}
A \textit{rearrangement protocol} reverts the rows and columns of an enhanced matrix (after we perform double scramble shuffles) back to their original positions so that we can reuse the cards without revealing them. This protocol was developed by Ruangwises and Itoh \cite{numberlink}, although slightly different protocols with the same idea were also used in other previous work \cite{makaro,revert1,revert2,ripple,sudoku}.

In the rearrangement protocol on an $m \times k$ enhanced matrix, first apply the double-scramble shuffle to the matrix. Then, turn over all marking cards in Row 0 and rearrange the columns such that each marking card with number $i$ will be in Column $i$. Analogously, turn over all marking cards in Column 0 and rearrange Rows $2,3,...,m$ accordingly.

\subsection{Neighbor Counting Protocol} \label{count}
Suppose we have an $m \times k$ matrix with each row encoding an integer in $\mathbb{Z}/k\mathbb{Z}$. A \textit{neighbor counting protocol} allows us to count the number of indices $i \geq 2$ such that Row $i$ encodes the same integer as Row 1, without revealing any other information. This protocol was developed by Ruangwises and Itoh \cite{numberlink}.

\begin{enumerate}
	\item Place marking cards to make the matrix become an $m \times k$ enhanced matrix.
	\item Apply the double-scramble shuffle.
	\item Turn over all encoding cards in Row 1. Locate the position of a \mybox{$\heartsuit$}. Suppose it is at Column $j$.
	\item Turn over all encoding cards in Column $j$. Count the number of \mybox{$\heartsuit$}s besides the one in Row 1. This is the number of indices that we want to know.
	\item Turn over all face-up cards. Apply the rearrangement protocol.
\end{enumerate}

\section{Verifying an Undirected Path} \label{path}
In this section, we will explain a \textit{path verification protocol}, which verifies the existence of an undirected path between vertices $s$ and $t$ in an undirected graph $G$. It is a special case $k=1$ of the protocol for the $k$ vertex-disjoint paths problem developed by Ruangwises and Itoh \cite{numberlink}.\footnote{Although the $k$ vertex-disjoint paths problem is NP-complete when $k$ is a part of the input, the special case $k=1$ in solvable in linear time. Hence, this protocol is actually unnecessary since $V$ can easily verifies existence of the path by him/herself given $G$. However, we explain the details of this protocol in order to show its idea, which will be modified and used in our main protocol in Section \ref{main}.}

We call $s$ and $t$ \textit{terminal vertices}, and other vertices \textit{non-terminal vertices}. We call a path $(v_1,v_2,...,v_\ell)$ \textit{minimal} if there are no neighboring vertices $v_i$ and $v_j$ such that $|i-j|>1$. Observe that given any path between $s$ and $t$, one can modify it to become a minimal one in linear time, so we can assume that $P$ knows a minimal path between $s$ and $t$.

Let $d$ be the maximum degree of a vertex in $G$. In linear time, we can color the vertices of $G$ with at most $d+1$ colors such that there are no neighboring vertices with the same color. This $(d+1)$-coloring is known to all parties.

On each terminal vertex $v$, $P$ publicly places a sequence $E_{d+2}(0)$. On each non-terminal vertex $v$ with the $x$-th color, $P$ secretly places a sequence $E_{d+2}(0)$ if $v$ is on $P$'s path, or a sequence $E_{d+2}(x)$ if $v$ is not on the path. Let $A(v)$ denote the sequence on each vertex $v$. Since the path is minimal, every non-terminal cell on the path has exactly two neighbors with a sequence encoding the same number as it (which is 0), while every terminal cell has exactly one such neighbor. On the other hand, every non-terminal cell not on the path has no neighbor with a sequence encoding the same number as it.

The idea is that, for every vertex $v$ with the $x$-th color, $P$ will add two ``artificial neighbors'' of $v$, both having $E_{d+2}(x)$ on it, and show that
\begin{enumerate}
	\item every non-terminal vertex $v$ (both on and not on the path) has exactly two neighbors with a sequence encoding the same number as $A(v)$, and
	\item every terminal vertex $v$ has exactly one neighbor with a sequence encoding the same number as $A(v)$.
\end{enumerate}

Formally, to verify each non-terminal (resp. terminal) vertex $v$ with the $x$-th color and with degree $d_v$, $P$ performs the following steps.
\begin{enumerate}
	\item Construct the following $(d_v+3) \times (d+2)$ matrix $M$.
	\begin{enumerate}
		\item In Row 1, place $A(v)$.
		\item In each of the next $d_v$ rows, place $A(v')$ for each neighbor $v'$ of $v$.
		\item In each of the last two rows, place $E_{d+2}(x)$.
	\end{enumerate}
	\item Apply the neighbor counting protocol to $M$. $V$ verifies that there are exactly two rows (resp. one row) encoding the same integer as Row 1.
	\item Put the sequences back to their corresponding vertices.
\end{enumerate}

If every vertex in $G$ passes the verification, then $V$ accepts.

\section{Verifying a Connected Spanning Subgraph} \label{main}
We get back to our main problem. Let $v_1,v_2,...,v_n$ be the vertices in $G$. In order to prove that $H$ is a connected spanning subgraph of $G$, it is sufficient to show that there is an undirected path between $v_i$ and $v_n$ in $H$ for every $i=1,2,...,n-1$.

Note that the path verification protocol in Section \ref{path} verifies a path between $s$ and $t$ in a graph $G$, where $G$ is known to all parties. In this section, we will modify that protocol so that it can verify a path between $s$ and $t$ in a subgraph $H$ of $G$, where $H$ is known to only $P$. Then, $P$ will perform the modified protocol for $n-1$ rounds, with $s=v_i$ and $t=v_n$ in each $i$-th round.

At the beginning, $P$ secretly places a sequence $B(e)$ on every edge $e \in G$ to indicate whether $e \in H$. ($B(e)$ is $E_2(1)$ if $e \in H$ and is $E_2(0)$ if $e \notin H$.) By doing this, the graph $H$ is committed and cannot be changed later.

Let $d$ be the maximum degree of a vertex in $G$. Like in the path verification protocol, consider a $(d+1)$-coloring, known to all parties, such that there are no neighboring vertices with the same color.

On every vertex $v$, $P$ publicly places a sequence $A_0(v)$, which is $E_{d+3}(d+2)$. $A_0(v)$ acts as a ``blank sequence'' guaranteed to be different from $A_1(v')$ on any vertex $v'$ during any round, which will be defined in the next step.

During each $i$-th round when $P$ wants to show that there is a path in $H$ between $s=v_i$ and $t=v_n$. First, $P$ selects a minimal path between $s$ and $t$ in $H$. On each terminal vertex $v$, $P$ publicly places a sequence $A_1(v)$, which is $E_{d+3}(0)$. On each non-terminal vertex $v$ with the $x$-th color, $P$ secretly places a sequence $A_1(v)$, which is $E_{d+3}(0)$ if $v$ is on the path and is $E_{d+3}(x)$ if $v$ is not on the path. Note that unlike $A_0(v)$ which remains the same throughout the whole protocol, $A_1(v)$ is changed in every round since it depends on the path selected in each round.

The verification steps are similar to the path verification protocol, except that in Step 1(b), $P$ first applies the sequence selection protocol in Section \ref{select} to determine whether to choose $A_0(v')$ or $A_1(v')$ for each neighbor $v'$ of $v$, depending on whether an edge $e$ between $v$ and $v'$ is in $H$ or not. The idea is that if $e \in H$, then $v'$ is still $v$'s neighbor in $H$, so $P$ chooses a sequence $A_1(v')$ and the rest works the same way as in the path verification protocol. On the other hand, if $e \notin H$, then $v'$ is not $v$'s neighbor in $H$, so $P$ chooses a sequence $A_0(v')$ which is guaranteed to be different from $A_1(v)$.

Formally, to verify each non-terminal (resp. terminal) vertex $v$ with the $x$-th color and with degree $d_v$, $P$ performs the following steps.
\begin{enumerate}
	\item Construct the following $(d_v+3) \times (d+3)$ matrix $M$.
	\begin{enumerate}
		\item In Row 1, place $A_1(v)$.
		\item For each neighbor $v'$ of $v$, let $e$ be an edge between $v$ and $v'$, and let $b$ be a bit encoded by $B(e)$. Apply the sequence selection protocol to choose a sequence $A_b(v')$ and place it in the next row of $M$. Repeatedly perform this for every neighbor of $v$ to fill the next $d_v$ rows.
		\item In each of the last two rows of $M$, place $E_{d+3}(x)$.
	\end{enumerate}
	\item Apply the neighbor counting protocol to $M$. $V$ verifies that there are exactly two rows (resp. one row) encoding the same integer as Row 1.
	\item Put the sequences back to their corresponding vertices.
\end{enumerate}

If every vertex in $G$ passes the verification, then $V$ accepts.

This protocol uses $2(d+3)(2n+2)+2d+2m$ encoding cards and $2d+5$ marking cards, where $n$ and $m$ are the numbers of vertices and edges of $G$, respectively, and $d$ is the maximum degree of a vertex in $G$. Therefore, the total number of required cards is $\Theta(dn)$.

\section{Proof of Correctness and Security}
We will prove the perfect completeness, perfect soundness, and zero-knowledge properties of our main protocol in Section \ref{main}.

\begin{lemma}[Perfect completeness] \label{lem1}
If $H$ is a connected spanning subgraph of $G$, then $V$ always accepts.
\end{lemma}

\begin{proof}
Suppose that $H$ is a connected spanning subgraph of $G$, then there exists a path between $v_i$ and $v_n$ in $H$ for every $i=1,2,...,n-1$.

First, we will prove the correctness of the sequence selection protocol in Section \ref{select}. Since $B$ encodes the number $b$, when placing $B$ in Row 2, the \mybox{$\heartsuit$} will be at Column $b+1$, the same column as the sequence $A_b$. After applying the pile-shifting shuffle, they will still be at the same column, so the sequence we get in Step 4 will be $A_b$.

Now consider the main protocol in each $i$-th round. In Step 1(b), $P$ always selects a sequence $A_1(v')$ if $e \in H$ and $A_0(v')$ if $e \notin H$. Since $A_0(v')$ is $E_{d+3}(d+2)$ and thus is different from $A_1(v)$, adding $A_0(v')$ to a new row of $M$ does not increase the number of rows encoding the same integer as Row 1. Therefore, the result will remain the same even if in Step 1(b) $P$ adds only the sequences on the vertices such that $e \in H$, which is equivalent to solely applying the path verification protocol in Section \ref{path} to verify a path between $v_i$ and $v_n$ on $H$.

The perfect completeness property of the path verification protocol has been proved in \cite{numberlink}, so we can conclude that $V$ always accepts.
\end{proof}

\begin{lemma}[Perfect soundness] \label{lem2}
If $H$ is not a connected spanning subgraph of $G$, then $V$ always rejects.
\end{lemma}

\begin{proof}
Suppose that $H$ is not a connected spanning subgraph of $G$, then there exists an index $i \in \{1,2,...,n-1\}$ such that there is no path between $v_i$ and $v_n$ in $H$. In Lemma \ref{lem1}, we have proved that the sequence selection protocol is correct, and the $i$-th round of the main protocol is equivalent to applying the path verification protocol to verify a path between $v_i$ and $v_n$ on $H$.

The perfect soundness property of the path verification protocol has been proved in \cite{numberlink}, so we can conclude that $V$ always rejects.
\end{proof}

\begin{lemma}[Zero-knowledge] \label{lem3}
During the verification, $V$ learns nothing about $H$.
\end{lemma}

\begin{proof}
To prove the zero-knowledge property, it is sufficient to prove that all distributions of the values that appear when the cards are turned face-up can be simulated by a simulator $S$ without knowing $H$.

\begin{itemize}
	\item In the sequence selection protocol:
	\begin{itemize}
		\item In Step 3, we turn over all cards in Row 2. This occurs right after a pile-shifting shuffle is applied to $M$. Hence, the \mybox{$\heartsuit$} has an equal probability to be at each of the $k$ columns, so this step can be simulated by $S$.
	\end{itemize}
	
	\item In the neighbor counting protocol:
	\begin{itemize}
		\item In Step 3, we turn over all encoding cards in Row 1. The order of Columns $1,2,...,k$ is uniformly distributed among all possible permutations due to the double-scramble shuffle. Hence, the \mybox{$\heartsuit$} has an equal probability to be at each of the $k$ columns, so this step can be simulated by $S$.
		\item In Step 4, we turn over all encoding cards in Column $j$. Suppose there are $t$ \mybox{$\heartsuit$}s besides the one in Row 1 ($t$ is now a public information). The order of Rows $2,3,...,m$ is uniformly distributed among all possible permutations due to the double-scramble shuffle. Hence, all $t$ \mybox{$\heartsuit$}s have an equal probability to be at each of the $\binom{m-1}{t}$ combinations of rows, so this step can be simulated by $S$.
	\end{itemize}
\end{itemize}

Therefore, we can conclude that $V$ learns nothing about $H$.
\end{proof}

\section{Applications to NP-Complete Problems} \label{app}
\subsection{Hamiltonian Cycle Problem} \label{hamiltonian}
Given an undirected graph $G$, determining whether $G$ has a Hamiltonian cycle (a cycle that visits each vertex exactly once) is known to be NP-complete \cite{book}. Suppose $P$ knows a Hamiltonian cycle $H$ of $G$ and wants to convince $V$ that $G$ has a Hamiltonian cycle without revealing any information about $H$.

To prove that $H$ is a Hamiltonian cycle of $G$, it is sufficient to show that
\begin{enumerate}
	\item $H$ is a connected spanning subgraph of $G$, and
	\item every vertex in $H$ has degree 2.
\end{enumerate}

At the beginning, $P$ commits $H$ by secretly placing a sequence $B(e)$ on every edge $e \in G$ to indicate whether $e \in H$. ($B(e)$ is $E_2(1)$ if $e \in H$ and is $E_2(0)$ if $e \notin H$.) The first condition can be verified by the protocol in Section \ref{main}.

To verify the second condition, $P$ first applies the copy protocol explained in Appendix \ref{copy} to make another copy of a sequence $B(e)$ on every edge $e$. (Each of the two copies will be used to verify each endpoint of $e$.) For each vertex $v \in H$, $P$ considers one (unused) copy of a sequence on every edge $e$ incident to $v$ and selects only the leftmost card of it (which is \mybox{$\clubsuit$} if $e \in H$ and is \mybox{$\heartsuit$} if $e \notin H$). Then, $P$ scrambles all selected cards together and turns over all of them, and $V$ verifies that there are exactly two \hbox{\mybox{$\clubsuit$}s} among them (which means $v$ has degree 2 in $H$). $V$ accepts if the verification passes for every vertex in $H$. This protocol also uses $\Theta(dn)$ cards.\footnote{There is an alternative way to verify a Hamiltonian cycle: $P$ publicly constructs an $n \times n$ adjacency matrix $M$ of $G$, then privately selects a permutation $\sigma$ and rearranges both the rows and columns of $M$ by $\sigma$. Finally, $P$ turns over all cards in the form $M(i,i+1)$ and $M(i,i-1)$ to show that they are all 1s. This protocol is simpler and more straightforward, but it requires $\Theta(n^2)$ cards, which is significantly greater than our protocol in sparse graphs.}

\subsection{Maximum Leaf Spanning Tree Problem} \label{tree}
Given an undirected graph $G$ and an integer $k$, the decision version of the maximum leaf spanning tree problem asks whether $G$ has a spanning tree with at least $k$ leaves (vertices with degree 1). This problem is also known to be NP-complete \cite{book}. Suppose $P$ knows a spanning tree $H$ of $G$ with at least $k$ leaves and wants to convince $V$ that the such tree exists without revealing any information about $H$.

To prove that $G$ has a spanning tree with at least $k$ leaves, it is sufficient to show that
\begin{enumerate}
	\item $H$ is a connected spanning subgraph of $G$, and
	\item $H$ has at least $k$ leaves.
\end{enumerate}
Note that it is not necessary to show that $H$ itself is a tree. (Even if $H$ itself is not a tree, any spanning tree of $H$ will also be a spanning tree of $G$, and every leaf of $H$ will still be a leaf of that tree, so $G$ must have a spanning tree with at least $k$ leaves.)

$P$ commits $H$ by the same way as in the Hamiltonian cycle problem, and uses the protocol in Section \ref{main} to verify the first condition.

To verify the second condition, $P$ makes an additional copy of every $B(e)$ like in the Hamiltonian cycle problem. For every vertex $v$, $P$ selects only the leftmost card of $B(e)$ on every edge $e$ incident to it, scramble these cards, and puts them into an envelope. (If there are less than $d$ cards, $P$ publicly adds more \mybox{$\heartsuit$}s until there are $d$ cards before scrambling them.) Then, $P$ scrambles all envelopes together. Next, $P$ picks an envelope, opens it and looks at the front side of all cards inside (without $V$ seeing the front side). If there is exactly one \mybox{$\clubsuit$} among them, $P$ reveals all cards to let $V$ verify that there is exactly one \mybox{$\clubsuit$} (which means the corresponding vertex is a leaf); otherwise, $P$ does not reveal the cards. $P$ repeatedly does this for every envelope. $V$ accepts if there are at least $k$ envelopes with exactly one \mybox{$\clubsuit$}. This protocol also uses $\Theta(dn)$ cards.

\subsection{Bridges Puzzle} \label{bridges}
Bridges, or the Japanese name Hashiwokakero, is a logic puzzle created by a Japanese company Nikoli, which also developed many other popular logic puzzles including Sudoku, Kakuro, and Numberlink.

A Bridges puzzle consists of a rectangular grid of size $p \times q$, with some cells called \textit{islands} containing an encircled positive number of at most 8. The objective of this puzzle is to connect some pairs of islands by straight lines called \textit{bridges} that can only run horizontally or vertically. There can be at most two bridges between each pair of islands, and the bridges must satisfy the following conditions \cite{nikoli} (see Fig. \ref{fig4}).
\begin{enumerate}
	\item \textit{Island condition}: The number of bridges connected to each island must equal to the number written on that island.
	\item \textit{Noncrossing condition}: Each bridge cannot cross islands or other bridges.
	\item \textit{Connecting condition}: The bridges must connect all islands into a single component.
\end{enumerate}

\begin{figure}
\centering
\begin{tikzpicture}
\draw[step=0.6cm,color={rgb:black,1;white,4}] (0,0) grid (4.2,4.2);

\node[draw,circle] at (0.3,0.3) {3};
\node[draw,circle] at (2.1,0.3) {4};
\node[draw,circle] at (3.9,0.3) {2};
\node[draw,circle] at (0.9,0.9) {2};
\node[draw,circle] at (2.7,1.5) {2};
\node[draw,circle] at (3.9,1.5) {3};
\node[draw,circle] at (0.9,2.1) {5};
\node[draw,circle] at (2.1,2.1) {6};
\node[draw,circle] at (3.3,2.1) {1};
\node[draw,circle] at (0.3,3.3) {2};
\node[draw,circle] at (2.1,3.3) {2};
\node[draw,circle] at (3.9,3.3) {1};
\node[draw,circle] at (0.9,3.9) {1};
\end{tikzpicture}
\hspace{1.2cm}
\begin{tikzpicture}
\draw[step=0.6cm,color={rgb:black,1;white,4}] (0,0) grid (4.2,4.2);

\node[draw,circle] at (0.3,0.3) {3};
\node[draw,circle] at (2.1,0.3) {4};
\node[draw,circle] at (3.9,0.3) {2};
\node[draw,circle] at (0.9,0.9) {2};
\node[draw,circle] at (2.7,1.5) {2};
\node[draw,circle] at (3.9,1.5) {3};
\node[draw,circle] at (0.9,2.1) {5};
\node[draw,circle] at (2.1,2.1) {6};
\node[draw,circle] at (3.3,2.1) {1};
\node[draw,circle] at (0.3,3.3) {2};
\node[draw,circle] at (2.1,3.3) {2};
\node[draw,circle] at (3.9,3.3) {1};
\node[draw,circle] at (0.9,3.9) {1};

\draw[very thick] (0.6,0.3) -- (1.8,0.3);
\draw[very thick] (2.4,0.3) -- (3.6,0.3);
\draw[very thick] (3.0,1.45) -- (3.6,1.45);
\draw[very thick] (3.0,1.55) -- (3.6,1.55);
\draw[very thick] (1.2,2.05) -- (1.8,2.05);
\draw[very thick] (1.2,2.15) -- (1.8,2.15);
\draw[very thick] (2.4,2.1) -- (3.0,2.1);
\draw[very thick] (2.4,3.3) -- (3.6,3.3);

\draw[very thick] (0.25,0.6) -- (0.25,3.0);
\draw[very thick] (0.35,0.6) -- (0.35,3.0);
\draw[very thick] (0.85,1.2) -- (0.85,1.8);
\draw[very thick] (0.95,1.2) -- (0.95,1.8);
\draw[very thick] (0.9,2.4) -- (0.9,3.6);
\draw[very thick] (2.05,0.6) -- (2.05,1.8);
\draw[very thick] (2.15,0.6) -- (2.15,1.8);
\draw[very thick] (2.1,2.4) -- (2.1,3.0);
\draw[very thick] (3.9,0.6) -- (3.9,1.2);
\end{tikzpicture}
\caption{An example of a Bridges puzzle (left) and its solution (right)}
\label{fig4}
\end{figure}
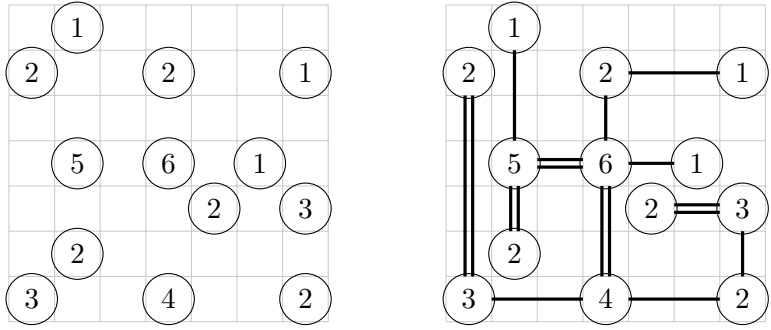

Determining whether a given Bridges puzzle has a solution has been proved to be NP-complete \cite{np}. Suppose $P$ knows a solution of the puzzle and wants to convince $V$ that it has a solution without revealing any information about the solution.

Define a \textit{lip} to be a line segment of a unit length on the Bridges grid that either separates two adjacent cells or lies on the outer boundary of the grid. For each lip $\ell$, let $b(\ell)$ be the number of bridges crossing through $\ell$ (including bridges coming out of the island from $\ell$ if $\ell$ is a lip of an island cell). First, $P$ secretly places on $\ell$ a sequence encoding $b(\ell)$ in $\mathbb{Z}/3\mathbb{Z}$. Then, $P$ publicly appends six \hbox{\mybox{$\clubsuit$}s} to the end of the sequence to make it encode $b(\ell)$ in $\mathbb{Z}/9\mathbb{Z}$ (while ensuring $V$ that $b(\ell)$ is at most 2). For each island cell $c$ with a number $n(c)$, $P$ publicly places a sequence encoding $n(c)$ in $\mathbb{Z}/9\mathbb{Z}$ on $c$.

For each cell $c$, let $b(\ell_1), b(\ell_2), b(\ell_3), b(\ell_4)$ be the numbers encoded by sequences on the top lip $\ell_1$, the right lip $\ell_2$, the bottom lip $\ell_3$, and the left lip $\ell_4$ of $c$, respectively (see Fig. \ref{fig5}). The steps of verifying $P$'s solution of the puzzle are as follows.

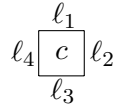
\begin{figure}
\centering
\begin{tikzpicture}
\draw[step=0.6cm] (0,0) grid (0.6,0.6);
\node[] at (0.3,0.3) {$c$};
\node[] at (0.34,0.8) {$\ell_1$};
\node[] at (0.84,0.3) {$\ell_2$};
\node[] at (0.32,-0.2) {$\ell_3$};
\node[] at (-0.2,0.3) {$\ell_4$};
\end{tikzpicture}
\caption{Positions of lips $\ell_1,\ell_2,\ell_3,\ell_4$ surrounding a cell $c$.}
\label{fig5}
\end{figure}

\begin{enumerate}
	\item For each lip $\ell$ located on the outer boundary of the Bridges grid, verify that $b(\ell)=0$ (no bridge goes beyond the grid), which can be shown by simply revealing the sequence on $\ell$.
	\item For each island cell $c$ with a number $n(c)$, verify that $b(\ell_1)+b(\ell_2)+b(\ell_3)+b(\ell_4) \equiv n(c)$ (mod 9) (the island condition).
	\item For each non-island cell $c$, verify that $b(\ell_1) \equiv b(\ell_3)$ (mod 9) and $b(\ell_2) \equiv b(\ell_4)$ (mod 9) (the number of bridges passing through $c$ is consistent), and also that $b(\ell_1) \cdot b(\ell_2) \equiv 0$ (mod 9) (the noncrossing condition).
\end{enumerate}

Steps 2 and 3 can be performed by applying a combination of copy and arithmetic protocols, which are explained in Appendix \ref{arithmatic}, and the neighbor counting protocol in Section \ref{count} (on a $2 \times 9$ matrix to verify the congruence).

Finally, construct a public graph $G$ with all islands being vertices of $G$, and two islands having an edge in $G$ if they are on the same row or column and there is no island between them (i.e. one can construct a valid bridge between them). Let $H$ be a private subgraph of $G$ such that two islands have an edge in $H$ if there is at least one bridge between them in $P$'s solution. $P$ performs the following steps to commit $H$ by placing a sequence $B(e)$, which is either $E_2(0)$ or $E_2(1)$, on every edge $e \in G$ to indicate whether $e \in H$.

\begin{enumerate}
	\item For each edge $e \in G$ with endpoints $u$ and $v$, consider any lip $\ell$ in the Bridges puzzle that lies between the two islands corresponding to $u$ and $v$.
	\item $P$ picks the leftmost card on $\ell$ and places it as a leftmost card of $B(e)$ without revealing it.
	\item $P$ shuffles the second and third leftmost cards on $\ell$ and looks at the front side of them (without $V$ seeing the front side). Then, $P$ selects a \mybox{$\clubsuit$} among them and turns it over to reveal the front side to $V$. (If both cards are \mybox{$\clubsuit$}s, $P$ can select any of them; if only one card is a \mybox{$\clubsuit$}, $P$ must select it.)
	\item $P$ places another unselected card in Step 3 as a rightmost card of $B(e)$ without revealing it.
\end{enumerate}

Observe that if there are one or two bridges between $u$ and $v$, then $B(e)$ will be $E_2(1)$; if there is no bridge between them, then $B(e)$ will be $E_2(0)$. Hence, these steps ensure that the subgraph $H$ is compatible with $P$'s solution of the puzzle without revealing any information about it.

Verifying the connecting condition is equivalent to verifying that $H$ is a spanning subgraph of $G$, which can be done by the protocol in Section \ref{main}. In total, this protocol uses $\Theta(pq)$ cards.

\section{Future Work}
We developed a physical card-based ZKP to verify the connected spanning subgraph condition, and showed applications of this protocol to verify solutions of three well-known NP-complete problems: the Hamiltonian cycle problem, the maximum leaf spanning tree problem, and the Bridges puzzle.

A possible future work is to explore methods to physically verify other NP-complete graph theoretic problems as well as other popular logic puzzles.

\appendix
\section{Copy and Arithmetic Protocols} \label{arithmatic}
In this appendix, we explain the copy and arithmetic protocols that can be used to verify problems in Section \ref{app}.

\subsection{Copy Protocol} \label{copy}
Given a sequence $A$ encoding an integer $a$ in $\mathbb{Z}/k\mathbb{Z}$, this protocol creates $m$ additional copies of $A$ without revealing $a$. It was developed by Shinagawa et al. \cite{polygon}.

\begin{enumerate}
	\item Reverse the $k-1$ rightmost cards of $A$, i.e. move each $(i+1)$-th leftmost card of $A$ to become the $i$-th rightmost card for $i=1,2,...,k-1$. This modified sequence, called $A'$, now encodes $-a$ (mod $k$).
	\item Construct a $(m+2) \times k$ matrix $M$ by placing the sequence $A'$ in Row 1 and a sequence $E_k(0)$ in each of Rows $2,3,...,m+2$.
	\item Apply the pile-shifting shuffle to $M$. Note that Row 1 of $M$ now encodes $-a+r$ (mod $k$), and other rows now encode $r$ (mod $k$) for a uniformly random $r \in \mathbb{Z}/k\mathbb{Z}$.
	\item Turn over all cards in Row 1 of $M$. Locate the position of a \mybox{$\heartsuit$}. Suppose it is at Column $j$.
	\item Shift the columns of $M$ cyclically to the left by $j-1$ columns. Turn over all face-up cards.
	\item The sequences in Rows $2,3,...,m+2$ of $M$ now encode $r-(-a+r) \equiv a$ (mod $k$), so we now have $m+1$ copies of $A$ as desired.
\end{enumerate}

\subsection{Addition Protocol} \label{add}
Given sequences $A$ and $B$ encoding integers $a$ and $b$ in $\mathbb{Z}/k\mathbb{Z}$, respectively. This protocol computes the sum $a + b$ (mod $k$) without revealing $a$ or $b$. It was developed by Shinagawa et al. \cite{polygon}.

\begin{enumerate}
	\item Reverse the $k-1$ rightmost cards of $A$. This modified sequence, called $A'$, now encodes $-a$ (mod $k$).
	\item Construct a $2 \times k$ matrix $M$ by placing $A'$ in Row 1 and $B$ in Row 2.
	\item Apply the pile-shifting shuffle to $M$. Note that Row 1 and Row 2 of $M$ now encode $-a+r$ (mod $k$) and $b+r$ (mod $k$), respectively, for a uniformly random $r \in \mathbb{Z}/k\mathbb{Z}$.
	\item Turn over all cards in Row 1 of $M$. Locate the position of a \mybox{$\heartsuit$}. Suppose it is at Column $j$.
	\item Shift the columns of $M$ cyclically to the left by $j-1$ columns. Turn over all face-up cards.
	\item The sequence in Row 2 of $M$ now encodes $(b+r)-(-a+r) \equiv a+b$ (mod $k$) as desired.
\end{enumerate}

\subsection{Multiplication Protocol} \label{multiply}
Given sequences $A$ and $B$ encoding integers $a$ and $b$ in $\mathbb{Z}/k\mathbb{Z}$, respectively, this protocol computes the product $a \cdot b$ (mod $k$) without revealing $a$ or $b$. It is a generalization of a protocol of Shinagawa and Mizuki \cite{triangle} to multiply two integers in $\mathbb{Z}/3\mathbb{Z}$.

\begin{enumerate}
	\item Repeatedly apply the copy protocol and the addition protocol to produce sequences $A_0,A_1,A_2,...,A_{k-1}$ encoding $0,a,2a,...,$ $(k-1)a$ (mod $k$), respectively.
	\item Apply the sequence selection protocol to select the sequence $A_b$ encoding $a \cdot b$ (mod $k$).
\end{enumerate}

\end{document}